\documentclass[journal]{IEEEtran}

\usepackage{newtxtext}         
\usepackage{newtxmath}         
\usepackage{helvet}            
\usepackage{courier}           
\usepackage{type1cm}           

\usepackage{amsmath,amsfonts}          
\usepackage{amssymb}                   
\usepackage{mathtools}                 
\usepackage{amsthm}                    
\usepackage{mathrsfs}                  

\usepackage{graphicx}                  
\DeclareGraphicsExtensions{.pdf,.png,.jpg,.eps}  
\usepackage[caption=false,font=normalsize,labelfont=sf,textfont=sf]{subfig}  

\usepackage{array}                     
\usepackage{booktabs}                  
\usepackage{threeparttable}            
\usepackage{multirow}                  
\usepackage{float}                     
\usepackage{stfloats}                  
\usepackage{multicol}                  
\usepackage[table]{xcolor}             

\usepackage{textcomp}                  
\usepackage{verbatim}                  
\usepackage{comment}                   
\usepackage{enumitem}                  
\usepackage{setspace}                  
\usepackage[bottom]{footmisc}          
\usepackage{ifthen}                    

\usepackage[noadjust]{cite}            
\usepackage{url}                       

\usepackage{makeidx}                   
\makeindex                             

\usepackage{algorithm}                 
\usepackage{algorithmic}               
\usepackage{longtable}  
\usepackage{ragged2e}   

\newcommand{\rr}{{\mathbb{R}}}          

\theoremstyle{definition}
\newtheorem{theorem}{Theorem}
\newtheorem{lemma}[theorem]{Lemma}      
\newtheorem{proposition}[theorem]{Proposition}  
\newtheorem{definition}{Definition}
\newtheorem{remark}{Remark}

\newtheorem{example}{Example}
\newtheorem{assumption}{Assumption}
\newtheorem{condition}{Condition}       

\allowdisplaybreaks                     
\begin{document}
\setstretch{0.94}

	\title{Decentralized and Equilibrium-Set-Oriented Stability Analysis and Control for Power Systems}
	\author{Peng Yang,~\IEEEmembership{Member,~IEEE},
		Liaoyuan Yang,
		and Feng Liu,~\IEEEmembership{Senior Member,~IEEE}%
	\thanks{P. Yang and L. Yang are with the School of Electronics and
	Information, Xi'an Polytechnic University, Xi'an 710048, China (e-mail:
	p-yang13@tsinghua.org.cn).}%
	\thanks{F. Liu is with the Department of Electrical Engineering, Tsinghua
	University, Beijing 100084, China (e-mail: lfeng@tsinghua.edu.cn).}}%

	\maketitle
	
	\begin{abstract}
		Conventional power-system stability analysis is largely centralized and centered on a single equilibrium point, which becomes increasingly restrictive in the presence of large-scale fluctuating renewable generation. This paper develops a decentralized framework for stability analysis and control that certifies the asymptotic stability of an equilibrium set rather than that of a given single operating point. To this end, we introduce a new notion termed input--output differential passivity (IODP), which decomposes equilibrium-set stability of the interconnected system into local requirements imposed on individual devices. These requirements are formulated without embedding a particular operating equilibrium into the local conditions; once the certified regions are constructed, stability verification for a given operating scenario reduces to checking whether its equilibrium lies in the certified set. The proposed conditions require each bus to possess a sufficient level of IODP, quantified by an IODP index. To compensate for an IODP shortage, we further develop an I/O-transformation-based passivation controller that reshapes the local input--output behavior of the corresponding device. In this way, all grid-connected components can be made to satisfy the decentralized conditions for system-wide stability. The proposed framework is validated on a modified IEEE 39-bus system. Simulation results demonstrate that it provides a scalable and equilibrium-set-oriented solution for stability certification and control under highly variable operating conditions.

	\end{abstract}

	\def\abstractname{Note to Practitioners}
	\begin{abstract}
		This paper addresses stability challenges caused by the large-scale integration of fluctuating renewable energy sources. Conventional stability-analysis frameworks are mainly centralized and tailored to a single operating equilibrium. As power systems become larger, more heterogeneous, and more variable, such frameworks may become increasingly difficult to use because of both computational burden and frequent equilibrium changes. The method proposed in this paper derives certified regions from local bus-level conditions only, which makes the assessment process more scalable for large systems. In addition, it certifies the stability of an equilibrium set rather than that of one specific equilibrium point. Therefore, when operating conditions change, the subsequent verification task is reduced to checking whether the new equilibrium remains inside the certified set, avoiding repeating system-wide analysis from scratch. To complement the analysis framework, we further design a local control strategy that adjusts the IODP characteristics of individual devices through a simple feedforward/feedback structure, so that the overall system can satisfy the proposed decentralized stability conditions.

	\end{abstract}


	\IEEEpeerreviewmaketitle

	\section{INTRODUCTION}
	\IEEEPARstart{S}{tability} has long been one of the most critical issues in power-system operation \cite{b58}, because its loss can trigger widespread blackouts and cause substantial socioeconomic damage \cite{b73}. Conventional stability analysis is largely grounded in a centralized, equilibrium-point-oriented paradigm. This paradigm emerged from traditional power systems dominated by a small number of large and controllable synchronous generators, and it has been highly effective in that context. Within this framework, methods such as eigenvalue analysis \cite{b1}, the Nyquist criterion \cite{b2}, and direct methods \cite{b3} all require a precise system-wide equilibrium point, which serves as the basis for model linearization or energy-function construction. As a result, stability assessment typically depends on centralized computations using system-wide information.

    This established paradigm, however, is being challenged by the ongoing transformation of power grids. The large-scale integration of fluctuating renewable generation, together with the massive deployment of heterogeneous inverter-interfaced devices, is fundamentally reshaping network structure and dynamic behavior \cite{b60,b65}. These changes increase operational complexity and introduce new stability challenges \cite{b52,b63}. In this setting, the conventional paradigm faces several fundamental limitations: it can impose heavy computation and communication burdens \cite{b66}, raise concerns about model privacy \cite{b67}, and, most importantly, rely on the availability of a fixed and explicitly known equilibrium point, which may be difficult to obtain under highly variable renewable operating conditions \cite{b68}. Consequently, ensuring stability in future power systems calls for a different analytical framework. This motivates the central question of this paper: how can a system-level stability property be certified through a decentralized framework that can accommodate a range of operating equilibria rather than a single prescribed one?
    


To address this question, prior studies have developed decentralized/distributed stability analysis methods for large-scale power systems. These methods can be broadly classified as linear, including positive realness \cite{b4}, Nyquist-based criteria \cite{b5}, and eigenvalue approximation \cite{b6}, and nonlinear, including passivity-based methods \cite{b27}, dissipative analysis \cite{b53}, input-to-state stability theory \cite{b8}, and sum-of-squares techniques \cite{b9}. However, most of them remain equilibrium-point-oriented, relying on linearization or explicit equilibrium information, and often requiring cross-subsystem data or centralized computation. Recent efforts include replacing equilibrium with synchronization \cite{b10,b11}, which still requires centralized computation, and matrix-phase-theory-based decentralized analysis \cite{b25}, which may be limited by the sectoriality assumption at low frequencies \cite{b36}. Along this direction, distributed control strategies have also been explored to enhance scalability and adaptability under varying operating conditions \cite{b75}.

Among existing decentralized stability frameworks, passivity stands out for interconnected nonlinear systems because of its compositional structure and scalability \cite{b37,b51}. Moreover, passivity naturally links stability analysis with controller synthesis and has been widely used in power-system control design \cite{b40,b43,b44}. However, most existing passivity-based methods still follow the classical equilibrium-point-oriented formulation, which requires explicit equilibrium information and is therefore less suitable under operating conditions with frequent equilibrium variations. This motivates a generalized passivity-based strategy that supports decentralized, equilibrium-set-oriented stabilization without embedding a particular equilibrium point into the local conditions. 

In this study, we build on the seminar work \cite{b35}, which introduced the concept of input--output differential passivity (IODP). Unlike classical passivity \cite{b37,b51}, IODP is formulated in terms of input and output differentials, thereby removing the explicit dependence of the dissipation inequality on a particular equilibrium point. This feature makes it well suited to decentralized stability analysis with local conditions that do not need to be tailored to a prescribed operating equilibrium. Our main contributions are summarized as follows.
\begin{itemize}
    \item \textbf{A decentralized stability-analysis framework for equilibrium sets.} Compared with the preliminary IODP formulation in \cite{b35}, this paper develops a more concrete stability-analysis framework by introducing the IODP index and corresponding certified regions for individual devices. On this basis, we establish an equilibrium-set-oriented criterion for system-wide stability in a decentralized form. Unlike conventional point-wise approaches, the resulting methodology derives local conditions without selecting a particular equilibrium point. Once the local certified regions are constructed, certifying a given operating scenario only requires checking whether the corresponding local equilibrium components lie in these regions. All required conditions are evaluated using local device models only, without centralized computation, which improves the scalability of stability assessment for large-scale grids under varying operating conditions.
    
\item \textbf{Decentralized passivation control based on the IODP index.} Beyond stability analysis, this paper develops an I/O-transformation-based control scheme that uses the IODP index to quantify and regulate the passivity shortage or excess of each device. This additional control layer turns the proposed framework from a certification tool into a constructive design method: devices that do not satisfy the required local conditions can be actively reshaped so as to enforce the conditions required for system-wide stability. The resulting method provides a scalable control solution that remains applicable under changing operating equilibria.
\end{itemize}

The remainder of this paper is organized as follows. Section II introduces the power-system model and formulates the problem. Section III develops the IODP framework and local characterization. Section IV establishes the IODP-based decentralized equilibrium-set stability theory. Section V presents the passivation-control design. Section VI reports case studies, and Section VII concludes the paper.

\textit{Notations}: We use bold letters to represent matrices, e.g., \(\boldsymbol{I}\) for the unit matrix. \(||\cdot||\) represents the Euclidean norm. \(\det(\cdot)\) denotes the determinant of a matrix. \(\mathrm{Int}(\cdot)\) denotes the interior of a set.
\(|\mathcal{V}|\) represents the cardinality of the set \(\mathcal{V}\), i.e., the number of elements in the set. For a set of vectors \(\{x_i\}_{i\in\mathcal{N}}\), the operator \(\mathrm{col}(x_i)_{i\in\mathcal{N}}\) denotes the stacking of these elements to form a column vector. $j=\sqrt{-1}$ denotes the imaginary unit.

\section{Problem Formulation}
This section presents the structure-preserving power system model and formulates our research question.
\subsection{Power System Model}
To provide a unified modeling foundation for decentralized analysis, this subsection formulates the interconnected power system as a graph-coupled dynamic--static DAE model. Specifically, a power system with $N$ buses and transmission lines is represented by an undirected graph $\mathcal{G} = (\mathcal{V}; \mathcal{L})$, where $\mathcal{V}$ and $\mathcal{L}$ denote the sets of buses and lines, respectively. Each bus is associated with a voltage \(V_i = (V_{Di}, V_{Qi})^T\in\rr^2\) and a current injection \(I_i = (I_{Di}, I_{Qi})^T\in\rr^2\), both defined in the common DQ coordinate. The device connected to the bus determines the relation between $V_i$ and $I_i$, and can be either dynamic (e.g., synchronous generators) or static (e.g., ZIP loads). Let $\mathcal{V}_1$ and $\mathcal{V}_2$ denote the sets of dynamic and static buses, respectively, with $\mathcal{V}_1\cap\mathcal{V}_2=\emptyset$ and $\mathcal{V}_1\cup\mathcal{V}_2=\mathcal{V}$.
\subsubsection{Dynamic Bus}
We describe the dynamic bus with a general input-state-output model
\begin{equation}\label{eq1}
    \left\{
    \begin{aligned}
        \dot{x}_i &= f_i(x_i, u_i) \\
        y_i &= h_i(x_i, u_i)
    \end{aligned}
    \right. \,,\forall i \in \mathcal{V}_1
\end{equation}
where \( x_i \in \mathbb{R}^{n_i} \) is the bus state variable, \( u_i \in \mathbb{R}^2 \) is the input variable, and \( y_i \in \mathbb{R}^2 \) is the output variable. The functions \( f_i: \mathbb{R}^{n_i} \times \mathbb{R}^2 \to \mathbb{R}^{n_i} \) and \( h_i: \mathbb{R}^{n_i} \times \mathbb{R}^2 \to \mathbb{R}^2 \) are continuously differentiable. We allow $u_i=V_i$, $y_i=-I_i$ or $u_i=-I_i$, $y_i=V_i$, depending on the specific device. 
\subsubsection{Static Bus}
We describe the static bus with a general input-output model
\begin{equation}\label{eq2}
    y_i = h_i(u_i)\,,\;\forall i \in \mathcal{V}_2,
\end{equation}
where \( h_i: \mathbb{R}^2 \to \mathbb{R}^2 \) is a continuously differentiable function. And $u_i=V_i$, $y_i=-I_i$ or $u_i=-I_i$, $y_i=V_i$, depending on the specific device.
\subsubsection{Power Network Coupling}
The bus voltages and currents satisfy Kirchhoff’s laws via the network admittance matrix \(\boldsymbol{Y} = \boldsymbol{G} + j\boldsymbol{B}\), where \(\boldsymbol{G}\) is the conductance matrix, and \(\boldsymbol{B}\) is the susceptance matrix. Defining aggregated vectors \(V_D = \text{col}(\{V_{Di}\}_{i\in\mathcal{V}})\), \(V_Q = \text{col}(\{V_{Qi}\}_{i\in\mathcal{V}})\), \(I_D = \text{col}(\{I_{Di}\}_{i\in\mathcal{V}})\), \(I_Q = \text{col}(\{I_{Qi}\}_{i\in\mathcal{V}})\), the network constraint is:
\begin{equation}
    I_D + jI_Q = \boldsymbol{Y}(V_D + jV_Q). \label{eq3}
    \end{equation}
This forms a static coupling subsystem \(H_{\text{net}}: u_{\text{net}} \mapsto y_{\text{net}}\) 
\begin{equation}
    \label{eq4}
    y_{\text{net}} = h_{\text{net}}(u_{\text{net}}),
\end{equation}
where we define $u_{\text{net}}:=\text{col}(\{y_{i}\}_{i\in\mathcal{V}})\in \mathbb{R}^{m}$, $y_{\text{net}}:=-\text{col}(\{u_{i}\}_{i\in\mathcal{V}})\in \mathbb{R}^{m}$, and \( m:= 2|\mathcal{V}|\). 

Combining each bus \eqref{eq1}, \eqref{eq2}, and the power network \eqref{eq4}, the entire power system is a feedback-interconnected system, as shown in Fig. \ref{fg1}. We can write its compact DAE form as
\begin{equation}\label{eq5}
\left\{
\begin{aligned}
    \dot{x} = f(x, u)\\
    0 = g(x, u)
\end{aligned}\right.,
\end{equation}
where \(x = \text{col}(\{x_i\}_{i\in\mathcal{V}_1})\) is the state variable, \(u \in \mathbb{R}^{m}\) is the algebraic variable, \(f = \text{col}(\{f_i\}_{i\in\mathcal{V}_1})\), and \(g(x, u) = h_{\text{net}}(h(x, u)) + u\), where \(h = \text{col}(\{h_i\}_{i\in\mathcal{V}})\).
\begin{figure}[!htb]
    \centering
    \includegraphics[width=5.5cm,height=3.3cm]{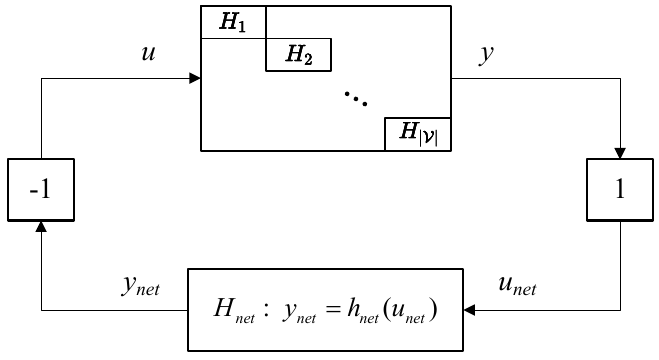}
    \caption{Feedback interconnection system schematic diagram.}
    \label{fg1}
\end{figure}
\begin{definition}\label{d3}
\((x^*,u^*)\) is called an equilibrium point of power system \eqref{eq5} if 
\[
\begin{cases}
0 = f(x^*, u^*) \\
0 = g(x^*,u^*).
\end{cases}
\]
\end{definition}
\begin{definition}\label{d4}
For a given domain \(\mathcal{D}_G\), the set of all equilibrium points of power system \eqref{eq5} contained in \(\mathcal{D}_G\) is denoted by
\[
\mathscr{E}:=\{(x,u)\in\mathcal{D}_G\mid f(x,u)=0,\ g(x,u)=0\}.
\]
We are interested in how to constructe \(\mathcal{D}_G\) such that $\mathscr{E}$ is asymptotically stable\footnote{The asymptotic stability of the set $\mathscr{E}$ means that for any $\varepsilon>0$, there exists $\delta>0$ such that $\text{dist}((x(0),u(0)),\mathscr{E})<\delta$ and $g(x(0),u(0))=0$ implies $\text{dist}\left((x(t),u(t)),\mathscr{E}\right)<\varepsilon$, $\forall t\geq0$, and that as $t\to\infty$, $\text{dist}\left((x(t),u(t)), \mathscr{E}\right)\to0$.}.
\end{definition}
The local component of an equilibrium point $(x^*,u^*)$---namely, $(x_i^*,u_i^*)$ for each $i\in\mathcal{V}_1$ and $u_i^*$ for each $i\in\mathcal{V}_2$---is referred to as the local equilibrium at bus $i$. We remark that power-system equilibria are generally non-unique because of the nonlinearities in $f$ and $g$.
\subsection{Problem Statement}
This paper addresses the following two objectives for the interconnected power system \eqref{eq5}: (i) establish an equilibrium-set-oriented and fully decentralized stability criterion such that the asymptotic stability of an equilibrium set \(\mathscr{E}\) can be certified without specifying any particular equilibrium point and by verifying only local bus-level models; and (ii) for buses that do not satisfy the decentralized conditions, develop a local passivation control to reshape their input--output relation so that they satisfy the proposed stability conditions, stabilizing the overall power system towards the equilibrium set \(\mathscr{E}\).

\section{IODP Framework and Local Characterization}
This section introduces the IODP framework used in this paper, including dynamic/static definitions, verifiable local conditions, and representative device-level properties.
\subsection{IODP Definitions for Dynamic and Static Buses}
We first define IODP for both the dynamic bus \eqref{eq1} and the static bus \eqref{eq2}. These bus-level properties will serve as the foundation of our decentralized equilibrium-set-oriented analysis. 
For the dynamic buses \eqref{eq1}, we define:
\begin{definition}[Dynamic IODP]\label{d1}
System \eqref{eq1} is said to be input-output differential passive in a domain \( \mathcal{D}_i\subset \mathbb{R}^{n_i} \times \mathbb{R}^2\) with IODP index \((\sigma_i,\rho_i)\), if there exists a continuously differentiable function \( S_i:\mathbb{R}^{n_i} \times \mathbb{R}^2 \to \mathbb{R} \) and three \( \mathcal{K} \)-class functions \( \alpha, \beta, \) and \( \gamma \) such that:
\begin{enumerate}
    \item For any \( (x_i,u_i) \in \mathcal{D}_i \), 
    \[
    \alpha(\|f_i(x_i,u_i)\|) \leq S_i(x_i,u_i) \leq \beta(\|f_i(x_i,u_i)\|);
    \]
    \item For any \( (x_i,u_i) \in \mathcal{D}_i \) and any \( \dot{u}_i \in \mathbb{R}^2 \),
    \begin{equation}\label{eq7}
    \begin{split}
    \frac{\partial S_i(x_i,u_i)}{\partial x_i} f_i(x_i,u_i) + \frac{\partial S_i(x_i,u_i)}{\partial u_i} \dot{u}_i & \leq\\\begin{bmatrix}
        \dot{u}_i\\\dot{y}_i
    \end{bmatrix}^{\mathrm{T}}\begin{bmatrix}
        -\sigma_i \boldsymbol{I}&\frac{1}{2}\boldsymbol{I}\\\frac{1}{2}\boldsymbol{I}&-\rho_i \boldsymbol{I}
    \end{bmatrix}\begin{bmatrix}
        \dot{u}_i\\\dot{y}_i
    \end{bmatrix}- \gamma(\|f_i(x_i,u_i)\|),
    \end{split}
    \end{equation}
    where \( \dot{y}_i := \frac{\partial h_i(x_i,u_i)}{\partial x_i} f_i(x_i,u_i) + \frac{\partial h_i(x_i,u_i)}{\partial u_i} \dot{u}_i \), \(\boldsymbol{I}\) is identity matrix.
\end{enumerate}
We denote such a system as $\mathrm{IODP}$\( (\sigma_i,\rho_i;\mathcal{D}_i) \). If \( \mathcal{D}_i = \mathbb{R}^{n_i} \times \mathbb{R}^2 \), the system is said to be globally IODP.
\end{definition}

For static buses, we define:
\begin{definition}[Static IODP]\label{d2}
Consider the static system \eqref{eq2}. Given a domain \( \mathcal{D}_i \subset \mathbb{R}^2 \), if for any \( u_i \in \mathcal{D}_i \) the following holds:
\begin{equation}\label{eq8}
    \begin{bmatrix}
        \boldsymbol{I}\\
        \frac{\partial h_i(u_i)}{\partial u_i}
    \end{bmatrix}^\mathrm{T}
    \begin{bmatrix}
        -\sigma_i \boldsymbol{I} & \frac{1}{2}\boldsymbol{I}\\
        \frac{1}{2}\boldsymbol{I} & -\rho_i \boldsymbol{I}
    \end{bmatrix}
    \begin{bmatrix}
        \boldsymbol{I}\\
        \frac{\partial h_i(u_i)}{\partial u_i}
    \end{bmatrix} \succeq 0,
\end{equation}
Then the system is said to be input-output differential passive in \( \mathcal{D}_i \) with IODP index \((\sigma_i,\rho_i)\), denoted as $\mathrm{IODP}$\( (\sigma_i,\rho_i;\mathcal{D}_i) \). If \( \mathcal{D}_i = \mathbb{R}^2 \), it is said to be globally IODP.
\end{definition}

The key distinction of IODP, as indicated by its name, is that both ports are described in differential form, as illustrated in Fig.~\ref{fg16}. This is reflected in the dissipation inequalities, where only differential variables are involved. The property is formulated without anchoring to any specific equilibrium, which enables equilibrium-set-oriented stability analysis.

The IODP index \(\sigma_i\) and \(\rho_i\) provide a quantitative measure of the differential passivity level. We remark that for a given subsystem and domain, the feasible IODP index pair is generally not unique: multiple \((\sigma_i,\rho_i)\) pairs may satisfy the IODP definition. They characterize the IODP margins of the differential input and differential output channels, respectively: positive indicates IODP excess, zero indicates critical IODP, and negative indicates IODP shortage. This channel-wise interpretation is also useful for control synthesis: one channel can be driven to the critical level (index set to zero), while the remaining shortage is absorbed by the other channel. As shown later, this strategy guarantees \(\sigma_i\rho_i=0\), which directly supports the solvability of the control design and broadens applicability to a wide class of IODP-shortage subsystems.
\begin{figure}[htb]
    \centering
    \includegraphics[width=0.45\textwidth]{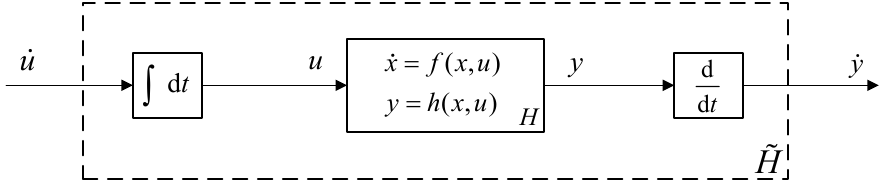}
    \caption{Input-Output differential passivity.}
    \label{fg16}
\end{figure}

\begin{remark}
Similar passivity concepts with differentiation include delta dissipativity \cite{b12} and Krasovskii passivity \cite{b13}. Our definition differs in that it introduces three \( \mathcal{K} \)-class functions to bind the storage function \( S \) and its derivative. Moreover, in Definition \ref{d1}, 1) and 2) are only required to hold within a specific domain \(\mathcal{D}\subset\mathbb{R}^n\times\mathbb{R}^m\), whereas the definition of delta dissipativity or Krasovskii passivity necessitates its validity for all possible \((x,u)\in\mathbb{R}^n\times\mathbb{R}^m\). The introduction of \(\mathcal{D}\) transforms the originally globally required conditions into locally sufficient ones, thereby broadening the applicability of the definition. In practical terms, most dynamic components in power systems can only satisfy local conditions.
\end{remark}

\subsection{Krasovskii-type Storage Function for IODP}
Building on Definition \ref{d1}, we now focus on the Krasovskii-type storage function class and derive its corresponding matrix-inequality characterization. This provides a theoretical reformulation of the IODP condition in terms of local differential information, as summarized in Proposition \ref{p4}.

\begin{proposition}
    \label{p4}Consider system \eqref{eq1}. If there exists a positive definite matrix $\boldsymbol{P}_i = \boldsymbol{P}_i^\mathrm{T} \succ 0$ and a real number $\epsilon_i > 0$ such that, for any $(x_i, u_i) \in \mathcal{D}_i$, the following holds:
\begin{equation}\label{eq28}
\begin{aligned}
&\begin{bmatrix}
\boldsymbol{P}_i\dfrac{\partial f_i}{\partial x_i} + \dfrac{\partial f_i^\mathrm{T}}{\partial x_i}\boldsymbol{P}_i + \epsilon_i \boldsymbol{I} & \boldsymbol{P}_i\dfrac{\partial f_i}{\partial u_i} \\
\dfrac{\partial f_i^\mathrm{T}}{\partial u_i}\boldsymbol{P}_i & \boldsymbol{0}
\end{bmatrix} \\
&- 
\begin{bmatrix}
\boldsymbol{0} & \boldsymbol{I} \\
\dfrac{\partial h_i}{\partial x_i} & \dfrac{\partial h_i}{\partial u_i}
\end{bmatrix}^\mathrm{T}
\begin{bmatrix}
-\sigma_i \boldsymbol{I} & \dfrac{1}{2}\boldsymbol{I} \\
\dfrac{1}{2}\boldsymbol{I} & -\rho_i \boldsymbol{I}
\end{bmatrix}
\begin{bmatrix}
\boldsymbol{0} & \boldsymbol{I} \\
\dfrac{\partial h_i}{\partial x_i} & \dfrac{\partial h_i}{\partial u_i}
\end{bmatrix} \preceq \boldsymbol{0}
\end{aligned}
\end{equation}
then the system satisfies $\text{IODP}(\sigma_i, \rho_i; \mathcal{D}_i)$.
\end{proposition}
\begin{proof}
Consider the Krasovskii-type storage function \(S_i(x_i,u_i) = f_i(x_i,u_i)^\mathrm{T}\boldsymbol{P}_if_i(x_i,u_i)\). Since \( \boldsymbol{P}_i \succ 0 \), we have
\[
\lambda_{\text{min}}(\boldsymbol{P}_i) \|f_i(x_i,u_i)\|^2 \leq S_i(x_i,u_i) \leq \lambda_{\text{max}}(\boldsymbol{P}_i) \|f_i(x_i,u_i)\|^2 
\]
where \( \lambda_{\text{min}}(\boldsymbol{P}_i) > 0 \) and \( \lambda_{\text{max}}(\boldsymbol{P}_i) > 0 \) denote the minimum and maximum eigenvalues of matrix \( \boldsymbol{P}_i \), respectively. Thus, \( S_i(x_i,u_i) \) satisfies 1) in Definition \ref{d1}.

On the other hand, calculate the derivative of \( S_i(x_i,u_i) \)
\[
\begin{aligned}
\dot{S}_i(x_i,u_i) &= \frac{\partial S_i(x_i,u_i)}{\partial x_i} \dot{x}_i + \frac{\partial S_i(x_i,u_i)}{\partial u_i} \dot{u}_i \\
&= \begin{bmatrix} \dot{x}_i \\ \dot{u}_i \end{bmatrix}^\mathrm{T} \begin{bmatrix} \boldsymbol{P}_i\frac{\partial f_i}{\partial x_i} + \frac{\partial f_i^\mathrm{T}}{\partial x_i}\boldsymbol{P}_i & \boldsymbol{P}_i\frac{\partial f_i}{\partial u_i} \\ \frac{\partial f_i^\mathrm{T}}{\partial u_i}\boldsymbol{P}_i & \mathbf{0} \end{bmatrix} \begin{bmatrix} \dot{x}_i \\ \dot{u}_i \end{bmatrix}
\end{aligned}
\]

From \eqref{eq28}, for any \( (x_i,u_i) \in \mathcal{D}_i \), we have
\begin{equation}\label{eq50}
\dot{S}_i(x_i,u_i) \leq 
\begin{aligned}
& \begin{bmatrix} \dot{x}_i \\ \dot{u}_i \end{bmatrix}^\mathrm{T} 
\begin{bmatrix} \mathbf{0} & \boldsymbol{I} \\ \frac{\partial h_i}{\partial x_i} & \frac{\partial h_i}{\partial u_i} \end{bmatrix}^\mathrm{T}
\begin{bmatrix} -\sigma_i \boldsymbol{I} & \frac{1}{2}\boldsymbol{I} \\ \frac{1}{2}\boldsymbol{I} & -\rho_i \boldsymbol{I} \end{bmatrix} \\
& \begin{bmatrix} \mathbf{0} & \boldsymbol{I} \\ \frac{\partial h_i}{\partial x_i} & \frac{\partial h_i}{\partial u_i} \end{bmatrix} 
\begin{bmatrix} \dot{x}_i \\ \dot{u}_i \end{bmatrix}
- \epsilon_i \|f_i(x_i,u_i)\|^2
\end{aligned}
\end{equation}
Note that
\[
\begin{bmatrix} \mathbf{0} & \boldsymbol{I} \\ \frac{\partial h_i}{\partial x_i} & \frac{\partial h_i}{\partial u_i} \end{bmatrix} \begin{bmatrix} \dot{x}_i \\ \dot{u}_i \end{bmatrix} = \begin{bmatrix} \dot{u}_i \\ \dot{y}_i \end{bmatrix} 
\]

Thus, from \eqref{eq50}, we get
\[
\dot{S}_i(x_i,u_i) \leq \begin{bmatrix} \dot{u}_i \\ \dot{y}_i \end{bmatrix}^\mathrm{T} \begin{bmatrix} -\sigma_i \boldsymbol{I} & \frac{1}{2}\boldsymbol{I} \\ \frac{1}{2}\boldsymbol{I} & -\rho_i \boldsymbol{I} \end{bmatrix} \begin{bmatrix} \dot{u}_i \\ \dot{y}_i \end{bmatrix} - \epsilon_i \|f_i(x_i,u_i)\|^2 
\]

Therefore, for any \( (x_i,u_i) \in \mathcal{D}_i \), \( S_i(x_i,u_i) \) satisfies 2) in Definition \ref{d1}, which completes the proof.
\end{proof}
Beyond its theoretical role, Proposition \ref{p4} also suggests a practical route to numerically compute IODP index. Since \eqref{eq28} must hold for all \((x_i,u_i)\in\mathcal{D}_i\), the problem is essentially a robust LMI condition. A common implementation is to fix \((x_i,u_i)\) (or a grid/sampled set in \(\mathcal{D}_i\)), so that \eqref{eq28} becomes an LMI in the decision variables \(\boldsymbol{P}_i\), \(\sigma_i\), and \(\rho_i\), which can be handled by SDP solvers. Conversely, for fixed \(\boldsymbol{P}_i\), \(\sigma_i\), and \(\rho_i\), one may scan \((x_i,u_i)\) to identify a certified region \(\mathcal{D}_i\).

\subsection{Representative Dynamic-Bus IODP Example}
To illustrate how the verifiable condition is applied in practice, this subsection presents a synchronous-generator example and computes its IODP region.
\begin{example}[IODP region for SG] Considering an SG modeled by the flux-decay 3rd-order dynamics
\begin{equation}\label{eq45}
\begin{cases}
\dot{\delta} = \omega \\
M\dot{\omega} = -D\omega - P^e + P^m + K_I\delta \\
T'_{d0}\dot{E}'_q = -E'_q + I_d(x_d - x'_d) + E_f
\end{cases}
\end{equation}
with the algebraic equation
\begin{equation}\label{eq46}
\begin{cases}
V_d = E'_q + x'_d I_q \\
V_q = -x_q I_d \\
P^e = E'_q I_d + (x'_d - x_q) I_d I_q
\end{cases}
\end{equation}
Here, $K_I \geq 0$ is the (optional) coefficient of frequency integral, which aims to balance the load variation. Table \ref{tab5} reports the parameters of the SG. 
\begin{table}[H]
  \centering
  \caption{Parameters of SG (values in p.u.).}
  \resizebox{\columnwidth}{!}{%
  \begin{tabular}{ccccccccc}
    \toprule
    $M$ & $D$ & $T'_{d0}$ & $x_d$ & $x_q$ & $x'_d$ & $P^m$ & $E_f$ & $K_I$ \\
    \midrule
    0.41 & 0.3 & 5.4 & 0.67 & 0.40 & 0.13 & 0.89 & 0.64 & 0.5 \\
    \bottomrule
  \end{tabular}
  }
  \label{tab5}
\end{table}
  According to Proposition 1, the input-output differential passivity of the SG can be determined by the linear matrix inequality \eqref{eq28}. Under the parameters of this example, the numerical solution obtained in MATLAB is as follows:
\begin{equation}\label{eq47}
    \begin{aligned}
\boldsymbol{P}_i &= \begin{bmatrix}
0.5658&0.0200&0.4116\\
    0.0200  &  0.0619 &  -0.0649\\
    0.4116  & -0.0649   & 3.0907\\
\end{bmatrix}, \\
\sigma_i &= -  8.5098, \quad \rho_i = 0
\end{aligned}
\end{equation}
Thus, SG satisfies \( \text{IODP}(-8.5098, 0; \mathcal{D}_i) \), where \( \mathcal{D}_i \) is a five-dimensional region obtained from \eqref{eq28}. To illustrate, we plot the two-dimensional cross-section (corresponding to \(V-(\delta-\theta)\), which is a subset of set \(\mathcal{D}_i\)) as illustrated in Fig. \ref{fg12}.
\begin{figure}[!t]
    \centering
    \includegraphics[width=0.33\textwidth]{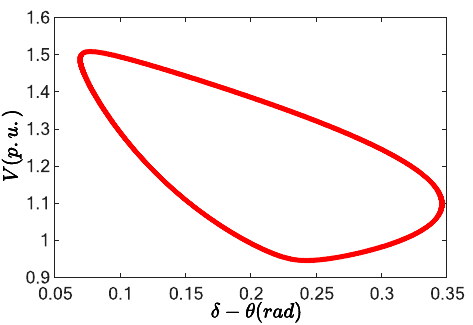}
    \caption{Cross-section of \(\mathcal{D}_i\) of the SG.}
    \label{fg12}
\end{figure}   
\end{example}
\subsection{IODP Properties of Static Subsystems}
This subsection summarizes key IODP properties of some representative static subsystems.
The seminar work \cite{b35} has shown that the power network itself is globally IODP as stated in the following proposition.

\begin{proposition}[\cite{b35}]\label{p1}
The power network \eqref{eq4} as a static subsystem is $\mathrm{IODP}$\( (0,0;\mathbb{R}^m) \).
\end{proposition}

Another typical static-bus model in power systems is the ZIP load, which is widely used to capture voltage-dependent load behavior. Consider the ZIP load at node $i$. Let $V_i$ denote the voltage magnitude, and let $P_i^L$ and $Q_i^L$ denote the active and reactive power consumptions, respectively. The ZIP load model is given by
\begin{equation}\label{eq51}
\begin{cases}
P_i^L = Z_p^l V_i^2 + I_p^l V_i + P_0^l \\
Q_i^L = Z_q^l V_i^2 + I_q^l V_i + Q_0^l
\end{cases},
\end{equation}
where $Z_p^l, Z_q^l, I_p^l, I_q^l, P_0^l,$ and $Q_0^l$ are nonnegative constants
corresponding to the constant impedance, constant current, and constant power
components of the load. It holds that
\begin{equation}\label{eq52}
\begin{cases}
-P_i^L = I_{Di} V_{Di} + I_{Qi} V_{Qi} \\
-Q_i^L = I_{Di} V_{Qi} - I_{Qi} V_{Di}
\end{cases},
\end{equation}
which yields
\begin{equation}\label{eq53}
\begin{cases}
I_{Di} = \dfrac{-P_i^L V_{Di} - Q_i^L V_{Qi}}{V_{Qi}^2 + V_{Di}^2} \\
I_{Qi} = \dfrac{-P_i^L V_{Qi} + Q_i^L V_{Di}}{V_{Qi}^2 + V_{Di}^2}
\end{cases},
\end{equation}
with $V_i = \sqrt{V_{Di}^2 + V_{Qi}^2}$.
Therefore, the ZIP load can be regarded as a static input--output mapping
$y_i = h_i(u_i)$ with input
$u_i := (V_{Di}, V_{Qi})^\top$ and output
$y_i := -(I_{Di}, I_{Qi})^\top$.

\begin{proposition}[IODP for ZIP load]\label{p5}
The ZIP load satisfies $\mathrm{IODP}(0,0;\mathcal{D}_i)$ in a certain
domain $\mathcal{D}_i$ defined as in \eqref{eq58}.
\end{proposition}
\begin{proof}
A direct calculation from \eqref{eq51} and \eqref{eq53} shows that the Jacobian
of $h_i(u_i)$ can be written as
\begin{equation}\label{eq54}
\frac{\partial h_i(u_i)}{\partial u_i}
= \boldsymbol{M}_Z
+ V_i^{-3} \boldsymbol{M}_I(u_i)
+ V_i^{-4} \boldsymbol{M}_P(u_i),
\end{equation}
where $\boldsymbol{M}_Z$ corresponds to the constant impedance:
\begin{equation}\label{eq55}
\boldsymbol{M}_Z :=
\begin{bmatrix}
Z_p^l & Z_q^l \\
-Z_q^l & Z_p^l
\end{bmatrix},
\end{equation}
$\boldsymbol{M}_I$ corresponds to the constant current:
	\begin{equation*}
        \boldsymbol{M}_I(u_i):=
            \begin{bmatrix}
            I_p^lV_{Qi}^2-I_q^lV_{Di}V_{Qi}
            &
            -I_p^lV_{Di}V_{Qi}+I_q^lV_{Di}^2
            \\
            -I_p^lV_{Di}V_{Qi}-I_q^lV_{Qi}^2
            &
            I_p^lV_{Di}^2+I_q^lV_{Di}V_{Qi}
            \end{bmatrix},
	\end{equation*}
and $\boldsymbol{M}_P$ corresponds to the constant power:
	\begin{equation*}
		\boldsymbol{M}_P(u_i):=\begin{bmatrix}
            aP_0^l-bQ_0^l & -bP_0^l+cQ_0^l\\
            -bP_0^l+cQ_0^l & cP_0^l+bQ_0^l
            \end{bmatrix},
	\end{equation*}
    where $a=V_{Qi}^2-V_{Di}^2$, $b=2V_{Qi}V_{Di}$, and $c=V_{Di}^2-V_{Qi}^2$.
Define the domain
\begin{equation}\label{eq58}
\mathcal{D}_i :=
\left\{
u_i \,\middle|\,
\boldsymbol{M}_Z + \boldsymbol{M}_Z^\top
+ \frac{\boldsymbol{M}_I + \boldsymbol{M}_I^\top}{V_i^{3}}
+ \frac{\boldsymbol{M}_P + \boldsymbol{M}_P^\top}{V_i^{4}} \succeq 0
\right\}.
\end{equation}
For all $u_i \in \mathcal{D}_i$, the symmetric part of the Jacobian
is positive semidefinite, which implies that the ZIP load satisfies
$\mathrm{IODP}(0,0;\mathcal{D}_i)$.

Moreover, noting that
\begin{equation}\label{eq59}
\boldsymbol{M}_Z + \boldsymbol{M}_Z^\top =
\begin{bmatrix}
2 Z_p^l & 0 \\
0 & 2 Z_p^l
\end{bmatrix}
\succeq 0,
\end{equation}
it follows that the constant impedance load is globally
$\mathrm{IODP}(0,0;\mathbb{R}^2)$.
\end{proof}

\section{IODP-Based Equilibrium-Set Stability Theory}
This section establishes how device-level IODP properties lead to system-level equilibrium-set stability for the interconnected DAE model. We first state decentralized stability conditions as local bus-level requirements and then prove that these conditions guarantee asymptotic stability of the equilibrium set.

\subsection{Decentralized Stability Conditions}
This subsection states decentralized stability requirements as bus-level IODP conditions.
\begin{condition}\label{c1}
Each dynamic bus \eqref{eq1} satisfies $\mathrm{IODP}$\( (0,0;\mathcal{D}_i) \) for a domain \( \mathcal{D}_i \subset \mathbb{R}^{n_i} \times \mathbb{R}^2 \) with a storage function $S_i(x_i,u_i)$.
\end{condition}

\begin{condition} \label{c2}
Each static bus \eqref{eq2} satisfies $\mathrm{IODP}$\( (0,0;\mathcal{D}_i) \) for a domain \( \mathcal{D}_i \subset \mathbb{R}^2 \).
\end{condition}

To connect these local conditions with the forthcoming global proof, we define several aggregate sets for the interconnected model. Let the algebraic manifold be
\begin{equation}\label{eq9}
    \mathbf{G} := \big\{(x,u) \in \mathbb{R}^n \times \mathbb{R}^m \mid 0 = g(x,u)\big\}.
\end{equation}
Let \( \mathcal{D}:= \mathcal{D}_1 \times \cdots \times \mathcal{D}_N \subset \mathbb{R}^n \times \mathbb{R}^m \) denote the Cartesian product of local IODP domains, and define \( \mathcal{D}_G := \mathcal{D} \cap \mathbf{G} \). Further define
\begin{equation}
    \label{eq13}
    \mathcal{D}_y\coloneqq \big\{y\in\mathbb{R}^m \mid y=h(x,u),(x,u)\in\mathcal{D}_G \big\},
\end{equation}
which is the image set of \(\mathcal{D}_G\) under \(h(x,u)\).

Condition \ref{c1} and \ref{c2} require dynamic and static buses satisfying IODP locally in \(\mathcal{D}_i\) with critical IODP indices. This enables a decentralized verification process: each bus checks its own local condition in \(\mathcal{D}_i\), and the system-level claim is then obtained from their interconnection structure. For a particular operating scenario, the remaining certification step is to compute its operating equilibrium and verify that \((x_i^*,u_i^*) \in \mathcal{D}_i\) for each dynamic bus and \(u_i^* \in \mathcal{D}_i\) for each static bus, equivalently \((x^*,u^*) \in \mathcal{D}_G\).

\subsection{From Local IODP to Equilibrium-Set Stability}
With the decentralized conditions established above, this subsection proves equilibrium-set stability for the interconnected system.
The following lemma shows that the summation of IODP storage functions can serve as a $\mathcal{W}$-function candidate \cite{b11}, which will be used to prove equilibrium-set stability.
\begin{lemma}\label{le:1}
Consider the interconnected power system \eqref{eq5}. If Condition \ref{c1} and \ref{c2} are satisfied, then for the function \( S(x,u) := \sum_{i\in\mathcal{V}_1} S_i(x_i,u_i) \), there exist three \( \mathcal{K} \)-class functions \( \alpha, \beta, \) and \( \gamma \) such that for any \( (x,u) \in \mathcal{D}_G \):
\begin{enumerate}
    \item \( \alpha(\|f(x,u)\|) \leq S(x,u) \leq \beta(\|f(x,u)\|); \)
    \item \( \dot{S} \leq -\gamma(\|f(x,u)\|). \)
\end{enumerate}
\end{lemma}
\begin{proof}
By Condition \ref{c1}, there exist \( \mathcal{K} \)-class functions \( \alpha_i \) and \( \beta_i \) for \( i\in\mathcal{V}_1\), such that \( \forall(x_i, u_i) \in \mathcal{D}_i \), the following holds: 
\[\sum_{i\in\mathcal{V}_1}\alpha_i(\|f_i\|)\leq S(x,u)=\sum_{i\in\mathcal{V}_1}S_i(x_i,u_i)\leq \sum_{i\in\mathcal{V}_1}\beta_i(\|f_i\|).\] 
For any \(r\in[0,\infty)\), define the set
\[\mathcal{D}(r):=\big\{(r_1,\dots,r_{|\mathcal{V}_1|})\in[0,\infty)^{|\mathcal{V}_1|}| \|(r_1,\dots,r_{|\mathcal{V}_1|})\|=r\big\}\] 
and functions
\[\alpha(r):=\underset{(r_1,...,r_{|\mathcal{V}_1|})\in \mathcal{D}(r)}{\mathrm{min}} (\alpha_1(r_1)+...+\alpha_{|\mathcal{V}_1|}(r_{|\mathcal{V}_1|}))\] 
\[\gamma(r):=\underset{(r_1,...,r_{|\mathcal{V}_1|})\in \mathcal{D}(r)}{\mathrm{min}} (\gamma_1(r_1)+...+\gamma_{|\mathcal{V}_1|}(r_{|\mathcal{V}_1|}))\] 
\[\beta(r):=\beta_1(r)+\dots+\beta_{|\mathcal{V}_1|}(r).\] 
Since \( \alpha_i \), \( \beta_i \), and \( \gamma_i \) are all  \( \mathcal{K} \)-class functions, it is easy to show that \( \alpha(r) \), \( \beta(r) \), and \( \gamma \) are also \(\mathcal{K} \)-class functions on \( [0, \infty) \). 
It holds that
\[\alpha(\|f\|)\leq\sum_{i\in\mathcal{V}_1}\alpha_i(\|f_i\|),\;\;  \gamma(\|f\|)\leq\sum_{i\in\mathcal{V}_1}\gamma_i(\|f_i\|)\]
and
\[\sum_{i\in\mathcal{V}_1}\beta_i(\|f_i\|)\leq\sum_{i\in\mathcal{V}_1}\beta_i(\|f\|)=\beta(\|f\|).\]
Thus, for any \((x,u)\in \mathcal{D}_G\), we have \(\alpha(\|f(x,u)\|)\leq S(x,u)\leq \beta(\|f(x,u)\|)\). 
It follows from Definition \ref{d1} that 
\[\dot{S}(x,u)=\sum_{i\in\mathcal{V}_1}\dot{S}_i(x_i,u_i)\leq \sum_{i\in\mathcal{V}_1}\dot{y}_i^\mathrm{T}\dot{u}_i-\sum_{i\in\mathcal{V}_1}\gamma_i(\|f_i(x_i,u_i)\|).\]
By Condition \ref{c2}, for $i\in\mathcal{V}_2$ we have $\dot{y}_i^\mathrm{T}\dot{u}_i\geq0$. Hence
\[\dot{S}(x,u)\leq \sum_{i\in\mathcal{V}}\dot{y}_i^\mathrm{T}\dot{u}_i-\sum_{i\in\mathcal{V}_1}\gamma_i(\|f_i(x_i,u_i)\|).\] 
Note that $u_{net}=y$, $y_{net}=-u$. Hence, it follows from Proposition \ref{p1} that \(\sum_{i\in\mathcal{V}}\dot{y}_i^\mathrm{T}\dot{u}_i=\dot{y}^\mathrm{T}\dot{u}=-\dot{y}_{net}^\mathrm{T}\dot{u}_{net}\leq0\). This leads to \( \dot{S} \leq -\gamma(\|f(x,u)\|)\).
\end{proof}

\begin{assumption} \label{as1}
The equilibrium set $\mathscr{E}$ in Definition~\ref{d4} is non-empty and bounded.
\end{assumption}
\begin{assumption} \label{as2}The power system \eqref{eq5} is algebraically non-singular within \( \overline{\mathcal{D}} \), i.e.,
\begin{equation}\label{eq11}
    \det\left(\frac{\partial g(x,u)}{\partial u}\right) \neq 0, \quad \forall (x,u) \in \overline{\mathcal{D}},
\end{equation}
where \( \overline{\mathcal{D}} \) is the closure of 
\( \mathcal{D} \).
\end{assumption}
\begin{remark}
Assumption~2 is a standard regularity condition for differential-algebraic systems and is widely adopted in power system stability analysis \cite{b70,b71,b72}. 
It guarantees, via the implicit function theorem, the local solvability of the algebraic equations $g(x,u)=0$ with respect to $u$, and implies that system \eqref{eq5} is an index-one DAE. Consequently, the system admits locally unique solution trajectories and a well-defined equilibrium set $\mathscr{E}$, making the stability analysis with respect to $\mathscr{E}$ mathematically well posed. Violation of this condition leads to singularity of the network algebraic equations. Solutions approaching the impasse surface are associated with short-term voltage instability in power systems, which is beyond the scope of this paper.
\end{remark}

The following theorem connects IODP to the equilibrium-set stability of the interconnected power system.
\begin{theorem}\label{th1}
Consider the power system \eqref{eq5} satisfying Assumption \ref{as1}-\ref{as2} and the decentralized Condition \ref{c1}-\ref{c2}. Then:
\begin{enumerate}
    \item The equilibrium set \( \mathscr{E} \) is asymptotically stable.
    \item If there exists an isolated equilibrium point within \( \mathscr{E} \), then every isolated equilibrium point within \( \mathscr{E} \) is asymptotically stable.
\end{enumerate}
\end{theorem}
\begin{proof}

 1). For any \(\epsilon>0\), choose \(r\in (0,\epsilon]\) such that
 \[B_r:=\{(x,u)\in\mathbf{G}|\mathrm{dist}((x,u),\mathscr{E})\leq r\}\subset \mathcal{D}_G.\]
 Let \(a:=\underset{\mathrm{dist}((x,u),\mathscr{E})=r}{\mathrm{min}}S(x,u)\). By Definition \ref{d1}, we have \(a>0\). For any \(b\in (0,a)\), consider the sublevel set \(S_b^{-1}=\{(x,u)\in\mathbf{G}|S(x,u)\leq b\}\) \footnote{The sublevel set \(S_b^{-1}\) in \(\mathbf{G}\) may contain multiple mutually disconnected branches, and only those branches that intersect with \(\mathcal{D}_G\) are considered. According to Assumption 1, such branches exist.}.   We have \(\mathscr{E}\subset S_b^{-1}\) and \(S_b^{-1}\) must be in \(B_r\) , namely \(S_b^{-1}\subset \mathrm{Int}(B_r)\). Otherwise, the boundary \(\partial B_r\) of \(B_r\) would intersect with \(S_b^{-1}\) at a point \((\hat{x},\hat{u})\). At this point, \(S(\hat{x},\hat{u})\geq a>b\), which contradicts \((\hat{x},\hat{u})\in S_b^{-1}\). On the other hand, by the continuity of \(S(x,u)\) and the fact that \(S(x,u)=0,\forall(x,u)\in\mathscr{E}\), it follows that there exists \(\delta>0\) such that
 \(\mathrm{dist}((x,u),\mathscr{E})\leq \delta \Rightarrow S(x,u)<b\)
 Hence \(\mathscr{E}\subset B_{\delta}\subset S_b^{-1}\subset \mathrm{Int}(B_r)\subset \mathcal{D}_G\). Then by \cite[Theorem 6]{b11}, \(S_b^{-1}\) is positively invariant and is an estimation of the \(f-\mathrm{RoA}\). That is \((x(0),u(0))\in S_b^{-1}\Rightarrow(x(t),u(t))\in S_b^{-1},\forall t\geq 0\) , and \(f(x(t),u(t))\rightarrow 0,t\rightarrow \infty\). Then we have \(\mathrm{dist}((x(0),u(0)),\mathscr{E})<\delta\Rightarrow \mathrm{dist}((x(t),u(t)),\mathscr{E})<r\leq \epsilon,\forall t\geq 0\). By Assumption 1, \(\mathscr{E}\) is bounded. Therefore, \(f(x(t),u(t))\rightarrow 0\Rightarrow \mathrm{dist}((x(t),u(t)),\mathscr{E})\rightarrow 0\). Thus, \(\mathrm{dist}((x(0),u(0)),\mathscr{E})<\delta\Rightarrow \mathrm{dist}((x(t),u(t)),\mathscr{E})\rightarrow0,t\rightarrow\infty\). As a result, \(\mathscr{E}\) is asymptotically stable.

 2). Let \((x^*,u^*)\in \mathscr{E}\) be an isolated equilibrium point. Due to the isolation of \((x^*,u^*)\), there exists a sufficiently small \(b>0\) such that a connected branch of the sublevel set \(S_b^{-1}\) contains \((x^*,u^*)\) and no other equilibrium points. Therefore, by the same reasoning as in the proof of 1), it follows that \((x^*,u^*)\) is asymptotically stable.
 \end{proof}
 

\section{IODP-Based Passivation Control Design}
In practical power systems, some dynamic subsystems may not naturally satisfy the IODP levels required by the decentralized stability conditions. This IODP shortage can prevent direct certification of system-level stability. To address this issue, this section proposes an IODP-based passivation control, which enables grid-connected components to achieve critical or even excessive IODP. We remark that the proposed control law is not restricted to two-dimensional input-output ports. It applies to general nonlinear subsystems with $u_i, y_i \in \mathbb{R}^m$. For clarity of notation, the derivation below is presented for the subsystem in \eqref{eq1} with $m=2$. The extension to general $m$ follows directly by replacing the corresponding dimensions from $2$ to $m$, without changing the core argument.
\subsection{I/O Transformation Passivation Control}
Consider the dynamic bus system \eqref{eq1}. Assume that there exist $\sigma_i \in \mathbb{R}$, $\rho_i \in \mathbb{R}$, and $\mathcal{D}_i \subset \mathbb{R}^{n_i} \times \mathbb{R}^2$ such that the system is $\text{IODP}(\sigma_i, \rho_i; \mathcal{D}_i)$. Suppose we want the system to be $\text{IODP}(\varsigma_i, \varrho_i; \widetilde{\mathcal{D}}_i)$, with $\varsigma_i > \sigma_i$ and $\varrho_i > \rho_i$, where $\widetilde{\mathcal{D}}_i$ is a new IODP region after controlling. 

For any invertible matrix $\boldsymbol{T} \in \mathbb{R}^{4 \times 4}$ and any constant vectors $u_{ic} \in \mathbb{R}^2$ and $y_{ic} \in \mathbb{R}^2$, consider a locally compensated subsystem whose external port variables $(\tilde{u}_i,\tilde{y}_i)$ are related to the original $(u_i,y_i)$ through the following affine transformation: 
\begin{equation}\label{eq16}
 \begin{bmatrix}
     \tilde{u}_i\\ \tilde{y}_i
 \end{bmatrix}\coloneqq \boldsymbol{T}\begin{bmatrix}
     u_i\\y_i
 \end{bmatrix}+\begin{bmatrix}
     u_{ic}\\y_{ic}
 \end{bmatrix} 
\end{equation}
The I/O transformation can be regarded as an extended input feedforward and output feedback control. Let the $2 \times 2$ block matrix of $\boldsymbol{T}$ be
\begin{equation}
    \label{eq17}\boldsymbol{T}=\begin{bmatrix}
        \boldsymbol{A}&\boldsymbol{B}\\\boldsymbol{C}&\boldsymbol{D}
    \end{bmatrix}
\end{equation}where $\boldsymbol{A} \in \mathbb{R}^{2 \times 2}$, $\boldsymbol{B} \in \mathbb{R}^{2 \times 2}$, $\boldsymbol{C} \in \mathbb{R}^{2 \times 2}$, and $\boldsymbol{D} \in \mathbb{R}^{2 \times 2}$. For the inverse matrix $\boldsymbol{T}^{-1}$ of $\boldsymbol{T}$, we also consider its $2 \times 2$ block matrix form, denoted as
\begin{equation}\label{eq18}
\boldsymbol{T}^{-1} = \begin{bmatrix}
\boldsymbol{E} & \boldsymbol{F} \\
\boldsymbol{G} & \boldsymbol{H}
\end{bmatrix},
\end{equation}
where $\boldsymbol{E} \in \mathbb{R}^{2 \times 2}$, $\boldsymbol{F} \in \mathbb{R}^{2 \times 2}$, $\boldsymbol{G} \in \mathbb{R}^{2 \times 2}$, and $\boldsymbol{H} \in \mathbb{R}^{2 \times 2}$.

Then, from \eqref{eq16}, the feedforward-feedback compensation can be written as follows.
\begin{equation}\label{eq19}
\begin{cases}
    u_i = \boldsymbol{E} \left( \tilde{u}_i - u_{ic} \right) + \boldsymbol{FC}u_i + \boldsymbol{FD}y_i \\
\tilde{y}_i = \boldsymbol{C}u_i + \boldsymbol{D}y_i + y_{ic}
\end{cases}
\end{equation}
In this paper, \eqref{eq19} is referred to as the I/O transformation control law, and its corresponding control block diagram is shown in Fig. \ref{fg7}. \(u_{ic}\) and \(y_{ic}\) are the bias terms of the input and output, which can be selected according to the desired steady-state relationship.
\begin{figure}[htb]
    \centering
    \includegraphics[width=0.98\columnwidth]{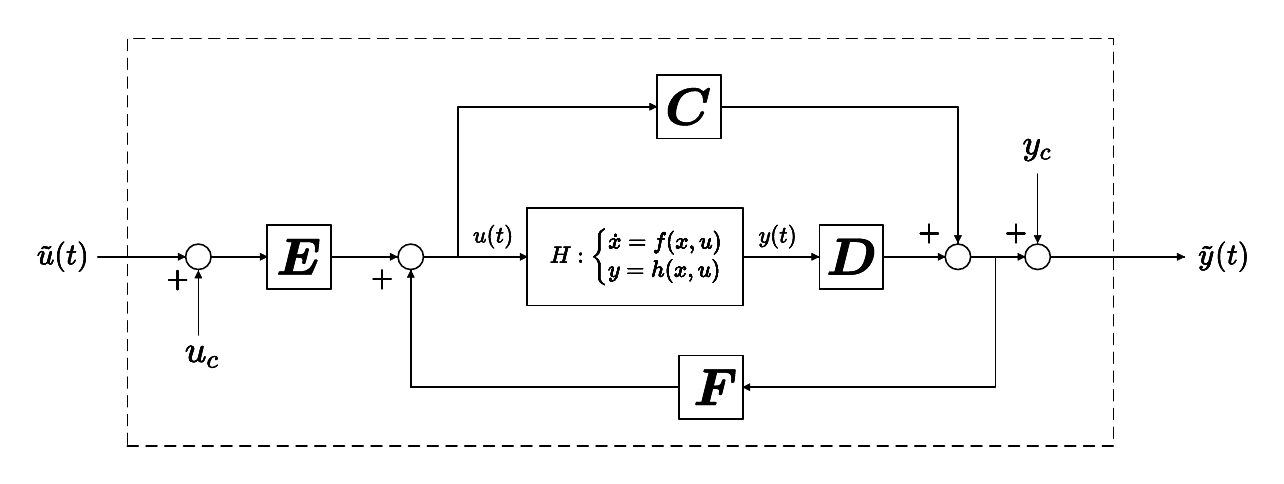}
    \caption{The block diagram of I/O transformation control.}
    \label{fg7}
\end{figure}

To make the I/O transformation control \eqref{eq19} well-defined, the following assumption is required.
\begin{assumption}\label{as3}
    Consider system \eqref{eq1} and the I/O transformation control \eqref{eq19}. For any $x_i \in \mathbb{R}^{n_i}$ and $\tilde{u}_i \in \mathbb{R}^2$, there exists a unique $u_i \in \mathbb{R}^2$ such that
    \begin{equation}\label{eq20}
        u_i =\boldsymbol{E} \left( \tilde{u}_i - u_{ic} \right) + \boldsymbol{FC}u_i + \boldsymbol{FD}h_i(x_i, u_i)
    \end{equation}
\end{assumption}
When Assumption \ref{as3} holds, \eqref{eq20} defines a function from $x_i, \tilde{u}_i$ to $u_i$, denoted as $u_i := v(x_i, \tilde{u}_i)$. Define the function
\begin{equation}
\tilde{f}_i(x_i, \tilde{u}_i) := f_i(x_i, v(x_i, \tilde{u}_i))
\label{eq21}
\end{equation}
\begin{equation}
\tilde{h}_i(x_i, \tilde{u}_i) := \boldsymbol{C}v(x_i, \tilde{u}_i) + \boldsymbol{D}h_i(x_i, v(x_i, \tilde{u}_i)) + y_{ic}
\label{eq22}
\end{equation}
Then, applying the control \eqref{eq19} to system \eqref{eq1}, we have
\begin{equation}
\begin{cases}
\dot{x}_i = \tilde{f}_i(x_i, \tilde{u}_i) \\
\tilde{y}_i = \tilde{h}_i(x_i, \tilde{u}_i)
\end{cases}
\label{eq23}
\end{equation}
Associated with the original IODP domain \(\mathcal{D}_i\), define the induced domain of the controlled system in the \((x_i,\tilde{u}_i)\)-coordinates as
\[
\widetilde{\mathcal{D}}_i := \big\{(x_i,\tilde{u}_i)\in\mathbb{R}^{n_i}\times\mathbb{R}^2\mid (x_i, v(x_i,\tilde{u}_i))\in\mathcal{D}_i \big\}.
\]

The following theorem shows that by properly choosing the matrix $\boldsymbol{T}$, the controlled system \eqref{eq23} can achieve prescribed IODP indices.
\begin{theorem}
\label{th4}
For given $\varsigma_i \in \mathbb{R}$, $\varrho_i \in \mathbb{R}$, $u_{ic} \in \mathbb{R}^2$, and $y_{ic} \in \mathbb{R}^2$, if Assumption \ref{as3} holds and there exists an invertible matrix $\boldsymbol{T} \in \mathbb{R}^{4 \times 4}$ such that
\begin{equation}
\boldsymbol{T}^\mathrm{T} \begin{bmatrix} -\varsigma_i \boldsymbol{I} & \frac{1}{2} \boldsymbol{I} \\ \frac{1}{2} \boldsymbol{I} & -\varrho_i \boldsymbol{I} \end{bmatrix} \boldsymbol{T} = \begin{bmatrix} -\sigma_i \boldsymbol{I} & \frac{1}{2} \boldsymbol{I} \\ \frac{1}{2} \boldsymbol{I} & -\rho_i \boldsymbol{I} \end{bmatrix}
\label{eq24}
\end{equation}
then the controlled system \eqref{eq23} satisfies $\text{IODP}(\varsigma_i, \varrho_i; \widetilde{\mathcal{D}}_i)$ with $\tilde{u}_i$ as the input and $\tilde{y}_i$ as the output.
\end{theorem}
\begin{proof}
  Given that the original input-output system \eqref{eq1} satisfies $\mathrm{IODP}(\sigma_i, \rho_i; \mathcal{D}_i)$, by Definition \ref{d1}, there exists a storage function $S_i(x_i, u_i)$ such that for any $(x_i, u_i) \in \mathcal{D}_i$,  
\(\alpha(\|f_i\|) \leq S_i \leq \beta(\|f_i\|)\)
and  
\begin{equation}\label{eq32}
\dot{S}_i(x_i, u_i) \leq \begin{bmatrix} \dot{u}_i \\ \dot{y}_i \end{bmatrix}^\mathrm{T} \begin{bmatrix} -\sigma_i \boldsymbol{I} & \frac{1}{2} \boldsymbol{I} \\ \frac{1}{2} \boldsymbol{I} & -\rho_i \boldsymbol{I} \end{bmatrix} \begin{bmatrix} \dot{u}_i \\ \dot{y}_i \end{bmatrix} - \gamma(\|f_i(x_i,u_i)\|)
\end{equation}  
Define the storage function of system \eqref{eq23} as:  
\begin{equation*}
\tilde{S}_i(x_i, \tilde{u}_i) := S_i(x_i, v(x_i, \tilde{u}_i))
\end{equation*}
For any \((x_i, \tilde{u}_i) \in \widetilde{\mathcal{D}}_i\), equivalently \((x_i, v(x_i, \tilde{u}_i)) \in \mathcal{D}_i\), we have
\begin{equation*}
\alpha(\|f_i(x_i, v(x_i, \tilde{u}_i))\|) \leq \tilde{S}_i(x_i, \tilde{u}_i) \leq \beta(\|f_i(x_i, v(x_i, \tilde{u}_i))\|)  
\end{equation*}
Thus, \(\tilde{S}_i\) satisfies 1) in Definition \ref{d1}.

Calculating the derivative of \(\tilde{S}_i\), it follows from \eqref{eq16} and \eqref{eq32} that for any \((x_i, \tilde{u}_i) \in \widetilde{\mathcal{D}}_i\), we have

\begin{equation*}
\begin{aligned}
\dot{\tilde{S}}_i(x_i, \tilde{u}_i) \leq &
\begin{bmatrix} 
\dot{u}_i \\ \dot{y}_i 
\end{bmatrix}^\mathrm{T} \boldsymbol{T}^{-\mathrm{T}} 
\begin{bmatrix} -\sigma_i \boldsymbol{I} & \frac{1}{2} \boldsymbol{I} \\ \frac{1}{2} \boldsymbol{I} & -\rho_i \boldsymbol{I} \end{bmatrix}
\boldsymbol{T}^{-1} \begin{bmatrix} \dot{u}_i \\ \dot{y}_i \end{bmatrix}\\
&- \gamma(\|f_i(x_i, v(x_i, \tilde{u}_i))\|) \\
= & \begin{bmatrix} \dot{\tilde{u}}_i \\ \dot{\tilde{y}}_i \end{bmatrix}^\mathrm{T} \begin{bmatrix} -\varsigma_i \boldsymbol{I} & \frac{1}{2} \boldsymbol{I} \\ \frac{1}{2} \boldsymbol{I} & -\varrho_i \boldsymbol{I} \end{bmatrix} 
\begin{bmatrix} \dot{\tilde{u}}_i \\ \dot{\tilde{y}}_i 
\end{bmatrix} - \gamma(\|\tilde{f}_i(x_i, \tilde{u}_i)\|) 
\end{aligned}
\end{equation*}

Thus, \(\tilde{S}_i\) satisfies 2) in Definition \ref{d1}, and therefore the controlled system \eqref{eq23} satisfies \(\mathrm{IODP}(\varsigma_i, \varrho_i; \widetilde{\mathcal{D}}_i)\).
\end{proof}

Theorem \ref{th4} shows that the proposed controller performs passivity shaping at the differential level: if an invertible matrix $\boldsymbol{T}$ satisfies \eqref{eq24} and Assumption \ref{as3} holds, then the controller \eqref{eq19} reshapes the original IODP level $(\sigma_i, \rho_i)$ of subsystem \eqref{eq1} into the target level $(\varsigma_i, \varrho_i)$ for the compensated closed-loop subsystem \eqref{eq23} on the induced domain \(\widetilde{\mathcal{D}}_i\). The bias terms $u_{ic}$ and $y_{ic}$ can be used to compensate constant offsets induced by the I/O transformation so that the desired steady-state operating point is preserved.

From an engineering viewpoint, $(\tilde{u}_i,\tilde{y}_i)$ are the external bus voltage and injected current of the compensated device seen by the network, whereas $(u_i,y_i)$ are internal signals of the original device model inside the local controller, as shown in Fig.\ref{fg7}. Hence, the physical meaning of the port variables are unchanged, but the closed-loop voltage-current relationship presented at that port is reshaped by the controller, similar to the widely used virtual-impedance control.

The following proposition presents the equivalent conditions of Assumption \ref{as3} under two special cases.
\begin{proposition}
\label{p2}
Consider system \eqref{eq1}. If \( y_i \) has an affine relationship with \( u_i \), i.e., \( y_i = h_{i1}(x_i) + h_{i2}(x_i)u_i \), then Assumption \ref{as3} holds iff for any \( x_i \in \mathbb{R}^{n_i} \), the matrix \( \boldsymbol{I} - \boldsymbol{FC} - \boldsymbol{FD}h_{i2}(x_i) \) is invertible. If \( y_i \) is independent of \( u_i \), i.e., \( y_i = h_i(x_i) \), then Assumption \ref{as3} holds iff the matrix \( \boldsymbol{I}-\boldsymbol{FC} \) is invertible.
\end{proposition}
\begin{proof}
Substituting the affine relationship into \eqref{eq19}, we obtain
\begin{equation}\label{eq25}
(\boldsymbol{I} - \boldsymbol{FC} - \boldsymbol{FD}h_{i2}(x_i))u_i = \boldsymbol{E}\left( \tilde{u}_i - u_{ic} \right) + \boldsymbol{FD}h_{i1}(x_i) 
\end{equation}
Therefore, Assumption \ref{as3} holds if and only if for any \( x_i \in \mathbb{R}^{n_i} \), the matrix \( \boldsymbol{I} - \boldsymbol{FC} - \boldsymbol{FD}h_{i2}(x_i) \) is invertible. In particular, the case where \( h_i \) is independent of \( u_i \) is equivalent to setting \( h_{i2}(x_i) = 0 \) in the above conclusion. Thus, Assumption \ref{as3} is equivalent to the matrix \( \boldsymbol{I} - \boldsymbol{FC} \) being invertible.
\end{proof}
Proposition \ref{p2} shows that when the function \( h_i \) satisfies a specific form, the originally implicit function-form assumption condition \eqref{eq19} can be transformed into an explicit function form, so that we can directly determine whether Assumption \ref{as3} holds by checking the invertibility of the matrix. 

\subsection{Constructive Solution of the I/O Transformation Matrix}
To enable practical controller synthesis, this subsection provides a constructive procedure for solving the I/O transformation matrix $\boldsymbol{T}$. Each sub-matrix $\boldsymbol{A} \in \mathbb{R}^{2 \times 2}$, $\boldsymbol{B} \in \mathbb{R}^{2 \times 2}$, $\boldsymbol{C} \in \mathbb{R}^{2 \times 2}$, and $\boldsymbol{D} \in \mathbb{R}^{2 \times 2}$ has $4$ undetermined coefficients. Substituting \eqref{eq17} into \eqref{eq24}, \eqref{eq24} is equivalent to the following polynomial equations
\begin{equation}\label{eq26}   
\begin{cases}
-\varsigma_i \boldsymbol{A}^\mathrm{T}\boldsymbol{A} + \frac{1}{2}\boldsymbol{A}^\mathrm{T}\boldsymbol{C} + \frac{1}{2}\boldsymbol{C}^\mathrm{T}\boldsymbol{A} - \varrho_i \boldsymbol{C}^\mathrm{T}\boldsymbol{C} = -\sigma_i \boldsymbol{I} \\
-\varsigma_i \boldsymbol{A}^\mathrm{T}\boldsymbol{B} + \frac{1}{2}\boldsymbol{A}^\mathrm{T}\boldsymbol{D} + \frac{1}{2}\boldsymbol{C}^\mathrm{T}\boldsymbol{B} - \varrho_i \boldsymbol{C}^\mathrm{T}\boldsymbol{D} = \frac{1}{2}\boldsymbol{I} \\
-\varsigma_i \boldsymbol{B}^\mathrm{T}\boldsymbol{B} + \frac{1}{2}\boldsymbol{B}^\mathrm{T}\boldsymbol{D} + \frac{1}{2}\boldsymbol{D}^\mathrm{T}\boldsymbol{B} - \varrho_i \boldsymbol{D}^\mathrm{T}\boldsymbol{D} = -\rho_i \boldsymbol{I}
\end{cases} 
\end{equation}
Equation set \eqref{eq26} contains $12$ scalar equations and $16$ unknowns, so it is underdetermined. A practical strategy is to fix $4$ variables first, which yields a square polynomial system with equal numbers of equations and unknowns. This determined system is multivariate polynomial and hence generally have finitely many roots in the complex field, which can then be computed efficiently using standard algebraic solvers.

In particular, if we consider diagonal matrices, i.e., $\boldsymbol{A} = a\boldsymbol{I}$, $\boldsymbol{B} = b\boldsymbol{I}$, $\boldsymbol{C} = c\boldsymbol{I}$, $\boldsymbol{D} = d\boldsymbol{I}$, where $a, b, c, d \in \mathbb{R}$ are undetermined coefficients, then \eqref{eq26} degenerates to
\begin{equation}\label{eq27}  
\begin{cases}
-\varsigma_i a^2 + ac - \varrho_i c^2 = -\sigma_i \\
-\varsigma_i ab + \frac{1}{2}ad + \frac{1}{2}cb - \varrho_i cd = \frac{1}{2} \\
-\varsigma_i b^2 + bd - \varrho_i d^2 = -\rho_i
\end{cases}
\end{equation}

The system \eqref{eq27} has four unknowns and three equations. Fixing any one of \(a, b, c, d\) reduces it to a determined quadratic system in three variables, which has at most eight complex solutions when multiplicities are counted.

Both \eqref{eq26} and \eqref{eq27} are underdetermined before fixing free variables, and thus admit infinitely many solutions. After fixing the free variables, only finitely many candidates remain. A feasible transformation matrix \(\boldsymbol{T}\) is then obtained by selecting a real solution and checking its invertibility; non-invertible candidates are discarded. If multiple invertible real matrices \(\boldsymbol{T}\) are available, any one can be used.

In general, verifying the existence of real roots is nontrivial. For the special case below, a sufficient condition can be derived.
\begin{proposition}\label{p3}
    When \(\varsigma_i = \varrho_i = 0\), there exist \(a, b, c, d \in \mathbb{R}\) satisfying \eqref{eq27} if \(\sigma_i\rho_i \leqslant \frac{1}{4}\).
\end{proposition}
\begin{proof}
Setting \(\varsigma_i = \varrho_i = 0\) in \eqref{eq27} yields

\begin{equation}\label{eq36}
\begin{cases}
ac = -\sigma_i \\
ad + cb = 1 \\
bd = -\rho_i
\end{cases}
\end{equation}

We consider four cases.

1) \(\sigma_i \neq 0\), \(\rho_i \neq 0\): From \eqref{eq36}, we have \(b \neq 0\) and \(c \neq 0\). Substituting \(a=-\sigma_i/c\) and \(d=-\rho_i/b\) into the second equation gives
\begin{equation}
\label{eq37}
c^2b^2 - cb + \sigma_i\rho_i = 0.
\end{equation}
Let \(z:=cb\). Then \eqref{eq37} becomes \(z^2-z+\sigma_i\rho_i=0\), which has real roots if
\begin{equation}
\label{eq38}
1-4\sigma_i\rho_i \ge 0 \iff \sigma_i\rho_i \le \frac{1}{4}.
\end{equation}

2) \(\sigma_i = \rho_i = 0\): A real solution is \(a=d=1\), \(b=c=0\). Hence \(\sigma_i\rho_i=0\le 1/4\).

3) \(\sigma_i = 0\), \(\rho_i \neq 0\): For any \(d\neq 0\), choosing \(a=1/d\), \(b=-\rho_i/d\), and \(c=0\) satisfies \eqref{eq36}. Again, \(\sigma_i\rho_i=0\le 1/4\).

4) \(\sigma_i \neq 0\), \(\rho_i = 0\): For any \(a\neq 0\), choosing \(d=1/a\), \(b=0\), and \(c=-\sigma_i/a\) satisfies \eqref{eq36}. Again, \(\sigma_i\rho_i=0\le 1/4\).

\end{proof}
Proposition \ref{p3} implies that a real solution of \eqref{eq36} exists whenever \(\sigma_i\rho_i \le 1/4\), even if the original subsystem is short of IODP. Moreover, because \(\sigma_i\) and \(\rho_i\) quantify the differential passivity margins associated with the two I/O channels, a practical synthesis strategy is to enforce one channel at the critical level (set one index to zero) and let the remaining shortage be absorbed by the other channel. This yields \(\sigma_i\rho_i=0\), which guarantees a real, invertible transformation matrix can be found; therefore, the method remains applicable to a broad class of IODP-shortage subsystems rather than only a narrow special case. 

\begin{example}[control for SG]
In Example 1, the SG satisfies \(\text{IODP}(-8.5098, 0; \mathcal{D}_1)\). This indicates insufficient input-output differential passivity in \(\mathcal{D}_1\), so the decentralized stability condition \ref{c1} is not met. Since \(\sigma_1\rho_1=0\), Proposition \ref{p3} guarantees the existence of an I/O transformation matrix \(\boldsymbol{T}_1\) satisfying \eqref{eq24} with \(\varsigma_i = \varrho_i = 0\). Solving \eqref{eq27} gives 

\begin{equation}\label{eq48}
   \boldsymbol{T}_1 = \begin{bmatrix}
 1.0000      &   0  &       0  &       0\\
         0   & 1.0000        & 0  &       0\\
    8.5098    &     0 &   1.0000    &     0\\
         0 &   8.5098   &   0 &   1.0000
\end{bmatrix}
\end{equation}

Moreover, Assumption \ref{as3} holds because \(y_i\) is affine in \(u_i\). Consequently, the controlled bus 1, viewed from its external port variables \((\tilde{u}_1,\tilde{y}_1)\), satisfies \(\text{IODP}(0,0;\widetilde{\mathcal{D}}_1)\), where
\[
\widetilde{\mathcal{D}}_1:=\big\{(x_1,\tilde{u}_1)\in\mathbb{R}^{n_1}\times\mathbb{R}^2\mid (x_1,v(x_1,\tilde{u}_1))\in\mathcal{D}_1\big\}.
\]
\end{example}




\section{Case Study}
This section validates the proposed theory and control framework on a modified IEEE 39-bus system with heterogeneous devices. To avoid ambiguity, the equilibrium information used in this section appears only at the scenario-certification stage: the local IODP conditions and certified regions are derived without fixing one nominal equilibrium, whereas each operating scenario is certified by checking whether its computed equilibrium lies in the corresponding certified region. We first present the test-system setup, then report the IODP analysis and passivation-control results, and finally verify stability through disturbance tests, random operating scenarios, and computational-time comparisons.

\subsection{System Description and Setup}
We consider the modified IEEE 39-bus benchmark system, as shown in Fig. \ref{fg8}. The original SGs at buses 30 and 31 are replaced by inverter-based sources with droop control (CD). The model of CD is \cite{b17}: 
\[
\begin{cases}
\tau_1 \dot{\theta}_i = -(\theta_i - \theta^{ref}) - d_1(P_i - P^{ref}) \\
\tau_2 \dot{V}_i = -(V_i - V^{ref}) - d_2(Q_i - Q^{ref})
\end{cases}
\]
where $V_i$ and $\theta_i$ are the magnitude and phase angle of the bus voltage in the $DQ$ reference frame; $P_i$ and $Q_i$ are the output active and reactive power; $P^{ref}$, $Q^{ref}$, $\theta^{ref}$, and $V^{ref}$ are the reference values; $\tau_1$ and $\tau_2$ are control time constants; $d_1$ and $d_2$ are droop coefficients.

Inverter-based sources are also integrated at buses 1, 7, 16, 18, 21, and 27, adopting either CD or virtual synchronous generator (VSG) control. The model of VSG is \cite{b18}:
\[
\begin{cases}
\dot{\theta}_i = \omega_i \\
M_i \dot{\omega}_i = -D_i \omega_i - P_i + P^{ref} + K_I \theta_i \\
T_i \dot{V}_i = K_Q (Q^{ref} - Q_i)/V_i
\end{cases}
\]
where $M_i$ is the virtual inertia; $T_i$ is the voltage control time constant; $K_Q$ is the reactive power control coefficient; $K_I\geq0$ is the optional frequency integral control coefficient.

The remaining 12 loads are constant impedance load buses. In addition to the above-mentioned power source and load buses, the system also includes 11 intermediate connection buses. Therefore, the system has a total of 16 dynamic buses and 23 static buses. 

The network and load parameters of the IEEE 39-bus system are taken from the MATPOWER toolkit, the parameters of SG1--SG8 are taken from \cite{b20}, and the parameters of inverter-based devices are listed in Table \ref{tab3}.

\begin{figure}[!t]
    \centering
    \includegraphics[width=0.45\textwidth]{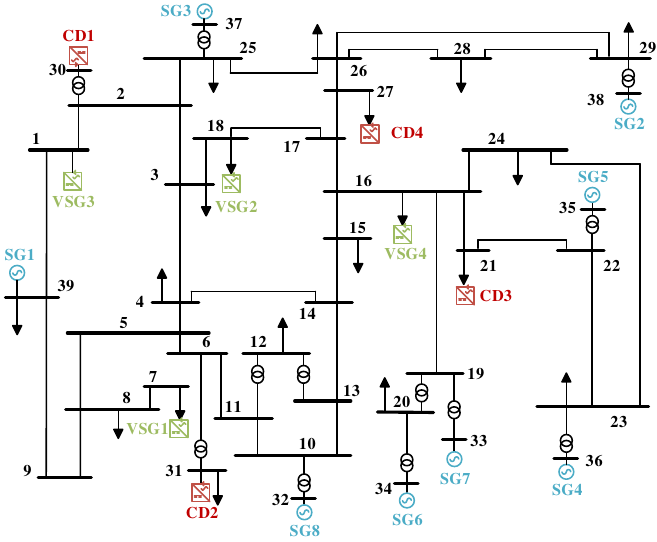}
    \caption{Schematic diagram of heterogeneous devices accessing the IEEE 39-bus system.}
    \label{fg8}
\end{figure}

\begin{table}[!htbp]
    \centering
    \caption{Parameters of inverter-based devices.}
    \begin{tabular}{c c c}
        \toprule
        Component & Parameter Name & Parameter Value (p.u.) \\
        \midrule
        CD1 & $\tau_1, \tau_2, d_1, d_2$ & 0.3, 8.1, 0.01, 0.01 \\ 
        CD2 & $\tau_1, \tau_2, d_1, d_2$ & 0.25, 9, 0.01, 0.01 \\
        CD3 & $\tau_1, \tau_2, d_1, d_2$ & 0.35, 8.4, 0.01, 0.01 \\
        CD4 & $\tau_1, \tau_2, d_1, d_2$ & 0.38, 7.2, 0.01, 0.01 \\
        VSG1 & $M, D, K_Q, K_I, T$ & 20, 0.5, 0.01, 0.5, 8 \\
        VSG2 & $M, D, K_Q, K_I, T$ & 14, 0.3, 0.01, 0.2, 7 \\
        VSG3 & $M, D, K_Q, K_I, T$ & 8, 0.2, 0.02, 0.4, 5 \\
        VSG4 & $M, D, K_Q, K_I, T$ & 10, 0.4, 0.02, 0.3, 7 \\
        \bottomrule
    \end{tabular}
    \label{tab3}
\end{table}

\subsection{IODP Analysis and Stability Verification}
Under the given parameters, the IODP indices $(\sigma_i, \rho_i)$ of each dynamic bus are computed using Proposition \ref{p4}. We remark again that the IODP index for the same device is non-unique. Here, column 2 of Table \ref{tab4} reports the index that minimizes $\sigma_i^2 + \rho_i^2$. If this pair does not satisfy $\sigma_i\rho_i \le 1/4$, we recompute a new index pair by solving Proposition \ref{p4} under the additional constraint $\rho_i=0$. The recomputed pair is reported in column 3 and then used in subsequent control design.

The I/O transformation control is applied to each dynamic device. To obtain the control law, we fix the free parameter $d=1$ and solve \eqref{eq27}; the resulting matrices are reported in column 4 of Table \ref{tab4}. By Theorem \ref{th4}, each controlled dynamic bus \eqref{eq23} satisfies the decentralized stabilization Condition \ref{c1}.

For static buses, the intermediate buses and constant impedance load buses naturally satisfy IODP$(0, 0; \mathbb{R}^2)$ by Propositions \ref{p1} and \ref{p5}. Therefore, under the designed control, all dynamic and static buses satisfy the proposed decentralized stability condition. By Theorem \ref{th1}, the resulting interconnected power system is thus predicted to be stable.

\begin{table*}[!t]
    \centering
    \footnotesize
    \caption{The IODP index and control parameters of each dynamic device.}
    \label{tab4}
    \resizebox{\textwidth}{!}{%
        \begin{tabular}{@{}lcccc@{}}
            \toprule
            \multirow{2}{*}{Device} & \multicolumn{2}{c}{Pre-control IODP Index} & I/O Transformation Control Parameters & Post-control IODP Index \\
            \cmidrule(lr){2-3} \cmidrule(lr){4-4} \cmidrule(lr){5-5}
            & Minimum-norm Pair$^\text{a}$ & Constrained Pair ($\rho_i=0$)$^\text{b}$ & $a, b, c, d$ & $\varsigma, \varrho$ \\
            \midrule
            SG1 & -0.5432, -0.5140 & -2638.0, 0 & 1, 0, 2638.0, 1 & 0, 0 \\
            SG2 & -2.0197, -0.8620 & -129.9007, 0 & 1, 0, 129.9007, 1 & 0, 0 \\
            SG3 & -1.4902, -0.8303 & -98.3183, 0 & 1, 0, 98.3183, 1 & 0, 0 \\
            SG4 & -1.9010, -0.9138 & -108.8303, 0 & 1, 0, 108.8303, 1 & 0, 0 \\
            SG5 & -1.8149, -0.8389 & -139.3291, 0 & 1, 0, 139.3291, 1 & 0, 0 \\
            SG6 & -2.0974, -0.9206 & -40.3053, 0 & 1, 0, 40.3053, 1 & 0, 0 \\
            SG7 & -2.6576, -1.0857 & -159.0196, 0 & 1, 0, 159.0196, 1 & 0, 0 \\
            SG8 & -2.1899, -0.9341 & -131.7683, 0 & 1, 0, 131.7683, 1 & 0, 0 \\
            CD1 & -0.0156, 0 & \textemdash & 1, 0, 0.0156, 1 & 0, 0 \\
            CD2 & -0.0149, 0 & \textemdash & 1, 0, 0.0149, 1 & 0, 0 \\
            CD3 & -0.0153, 0 & \textemdash & 1, 0, 0.0153, 1 & 0, 0 \\
            CD4 & -0.0153, 0 & \textemdash & 1, 0, 0.0153, 1 & 0, 0 \\
            VSG1 & -0.4063, -0.1576 & \textemdash & 0.9312, 0.1576, 0.4363, 1 & 0, 0 \\
            VSG2 & -0.5496, -0.4276 & \textemdash & 0.6225, 0.4276, 0.8828, 1 & 0, 0 \\
            VSG3 & -1.1519, -0.3010 & -1.2602, 0 & 1, 0, 1.2602, 1 & 0, 0 \\
            VSG4 & -0.6025, -0.4853 & -6.1644, 0 & 1, 0, 6.1644, 1 & 0, 0 \\
            \bottomrule
        \end{tabular}%
    }

    \vspace{0.5ex}
    {\raggedright\footnotesize
    $^\text{a}$ Minimum-norm IODP index pair $(\sigma_i, \rho_i)$, obtained by minimizing $\sigma_i^2 + \rho_i^2$.\\
    $^\text{b}$ Recomputed IODP indices under the additional constraint $\rho_i=0$ if the minimum-norm pair does not satisfy $\sigma_i\rho_i \le 1/4$.\par}
\end{table*}

To validate this theoretical prediction using an independent and widely accepted approach, a time-domain simulation is conducted under a prescribed disturbance. Specifically, at $t = 5$ s, a step increase of 5\% in both active and reactive loads is applied to all load buses. This disturbance introduces transient dynamics and drives the system away from its pre-disturbance operating condition.

The corresponding responses are shown in Fig.~\ref{fg19}. It can be observed that all trajectories remain bounded during the transient and eventually settle to new steady-state values after the disturbance.
These results indicate that the system remains stable under the applied disturbance, which justifies our theory.

\begin{figure}[H]
    \centering
    \includegraphics[width=1.0\columnwidth]{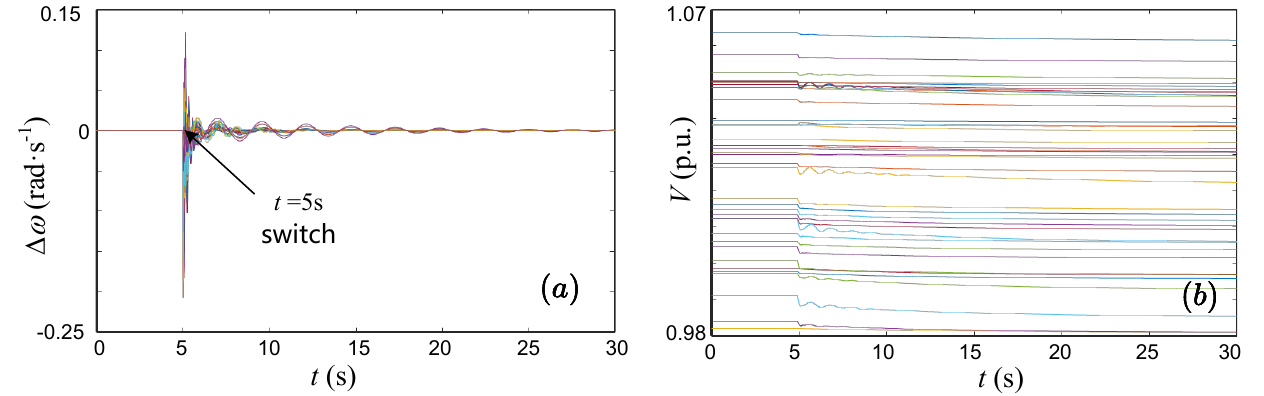}
    \caption{Time-domain responses under a 5\% step increase in active/reactive loads at all load buses: (a) frequency trajectories of generator buses; (b) voltage magnitudes at representative buses.}
    \label{fg19}
\end{figure}

\subsection{Random Operating Scenarios}

A core feature of the proposed method is that the local certificate is not anchored to a single nominal equilibrium: the IODP regions are constructed for an equilibrium set and can therefore accommodate changes in operating points. For any sampled operating scenario, we check whether the corresponding local equilibrium components lie inside the certified IODP regions. If this membership test is satisfied, the scenario is certified as asymptotically stable by our method. To illustrate this mechanism, we consider five load-fluctuation ranges and generate 1000 random scenarios for each range. For a given level $\alpha$, the net loads at each bus are scaled by independently sampled factors from the uniform interval $[1-\alpha,\,1+\alpha]$, leading to different equilibrium points. For each equilibrium, we apply the above membership test and compare the result with the centralized Jacobian-eigenvalue analysis. 

Whenever the proposed method certifies an operating point as stable, the result is fully consistent with the Jacobian-eigenvalue analysis, which supports the correctness of the proposed certification framework. At the same time, the method is conservative in some scenarios: an operating point outside the IODP region can still be stable according to the Jacobian test. This observation is consistent with our theory because the proposed condition is sufficient rather than necessary. The percentages of such conservative scenarios under different fluctuation levels are summarized in Table~\ref{table5}. Overall, the conservatism is negligible under normal operating-condition variations and becomes noticeable only under relatively severe fluctuations, which still supports the practical value of the method.

To visualize the equilibrium points and IODP regions, consider the $\alpha=\pm10\%$ fluctuation level. Fig.~\ref{fg17}(a) shows the projection of the VSG equilibrium points at bus 18 onto the $V$--$\theta$ cross-section of its IODP region. Fig.~\ref{fg17}(b) shows the projection of the SG equilibrium points at bus 32 onto the $V$--$(\delta-\theta)$ cross-section of its IODP region.
\begin{figure}[!htbp]
    \centering
    \includegraphics[width=1.0\columnwidth]{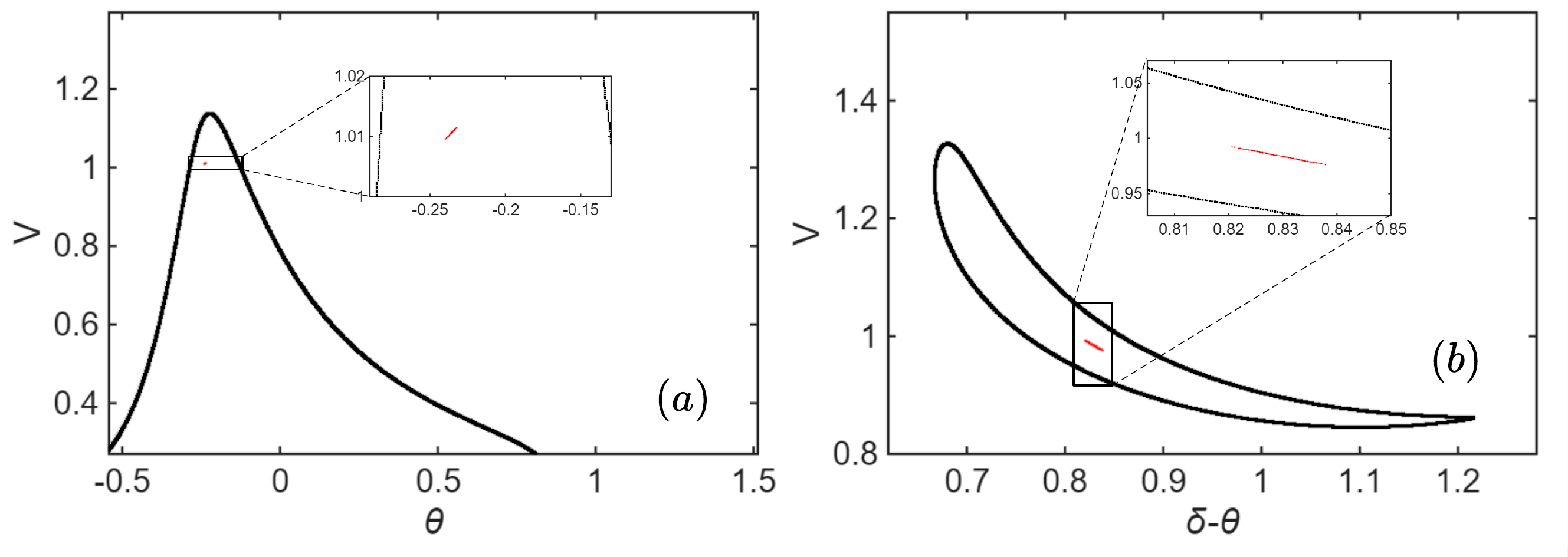}
    \caption{Location of equilibrium points in 1000 random scenarios of VSG and SG under $\pm10\%$ load fluctuations.}
    \label{fg17}
\end{figure}

\subsection{Computational Time Comparison}

To further highlight the computational advantage of the proposed method, we compare the runtime of the centralized eigenvalue-based approach with that of the decentralized IODP-based approach over the 1000 random operating scenarios. Because the local checks in the proposed approach can be performed in parallel, we define the decentralized runtime as the longest local computation time required to test whether a local equilibrium $(x_i^*,u_i^*)$ lies in the corresponding IODP region $\mathcal{D}_i$. This test is implemented by computing the eigenvalues of the left-hand side of \eqref{eq28}. For the centralized approach, the runtime is defined as the time required to compute the eigenvalues of the system Jacobian at the equilibrium $(x^*,u^*)$.

The total computation time over 1000 random operating scenarios is summarized in Table~\ref{table5} for each different load-fluctuation level. All computation is performed in Matlab 2025a on the same PC with an AMD Ryzen 7 7840H CPU (3.80 GHz).

\begin{table}[!htbp]
    \centering
    \caption{Conservative scenarios and total computation times under different load-fluctuation levels.}
    \scriptsize
    \begin{tabular}{c c c c}
        \toprule
        \multirow{2}{*}{Fluctuation level} & \multirow{2}{*}{Conservative scenarios} & \multicolumn{2}{c}{Total computation time (s)} \\
        \cmidrule(lr){3-4}
         &  & Decentralized & Centralized \\
        \midrule
        $\pm10\%$ & $0\%$ & 0.38 & 4.06 \\
        $\pm20\%$ & $0\%$ & 0.42 & 4.40 \\
        $\pm30\%$ & $0\%$ & 0.42 & 4.41 \\
        $\pm35\%$ & $0.2\%$ & 0.44 & 4.64 \\
        $\pm40\%$ & $3.4\%$ & 0.42 & 4.25 \\
        $\pm45\%$ & $10.6\%$ & 0.37 & 3.98 \\
        \bottomrule
    \end{tabular}
    \label{table5}
\end{table}

The centralized method is consistently much slower than the proposed decentralized method across all tested scenarios. According to Table~\ref{table5}, its total runtime is approximately 10.12--10.76 times that of the decentralized method.

This gap is also consistent with the underlying computational complexity. In Matlab, computing the eigenvalues of a matrix $A$, i.e., `eig(A)`, is implemented through QR-based routines whose dominant complexity is of order $O(n^3)$~\cite{b74}. In the centralized method, $A$ is the 44-dimensional system Jacobian associated with all dynamic-device states in the IEEE 39-bus system. By contrast, in the proposed decentralized method, verifying \eqref{eq28} only requires the eigenvalues of a matrix built from the state and input variables of a single device, whose dimension is only 4 or 5. The dominant complexity is therefore reduced from approximately $O(44^3)$ to $O(4^3)$ or $O(5^3)$. Moreover, as the system size grows, the dimension of the centralized matrix increases accordingly, whereas that of the local test remains unchanged unless higher-order device models are introduced.


       

\section{Conclusion}
This paper developed a decentralized and equilibrium-set-oriented framework for stability analysis and control of power systems with heterogeneous devices. By introducing input--output differential passivity (IODP), the system-level stability requirement is decomposed into local bus-level conditions that do not depend on a specific operating equilibrium. For any given operating scenario, certification is then reduced to checking whether the corresponding local equilibrium components lie in the certified IODP regions. A Krasovskii-type characterization, together with the proposed I/O-transformation-based passivation controller, provides a constructive mechanism for enforcing these conditions on devices with insufficient IODP. Case studies on a modified IEEE 39-bus system show that the proposed method can consistently certify stability under varying operating points while substantially reducing computational burden through decentralized local verification. This feature is particularly attractive for renewable-rich power grids with many heterogeneous devices and frequently changing operating conditions.
Future work will focus on reducing the conservatism of the certification conditions and developing more practical controller implementations for dynamic devices.




\bibliographystyle{IEEEtran}
\bibliography{IEEEabrv}
%
	
\end{document}